\documentclass[hidelinks]{article}
\usepackage[a4paper, total={6in, 8in}]{geometry}

\usepackage[utf8]{inputenc}
\usepackage[T1]{fontenc}

\usepackage{lineno,hyperref}
\modulolinenumbers[5]
\usepackage{stmaryrd}
\usepackage{xcolor}
\usepackage{bbold}
\usepackage{subfig}
\usepackage{enumitem}
\usepackage{tabularx}
\usepackage{booktabs}

\usepackage{amsmath,amsthm,amssymb}
\newtheorem{Theorem}{Theorem}
\newtheorem{Definition}{Definition}

\usepackage{cleveref}
\usepackage{makecell, multirow}
\usepackage{subcaption}
\usepackage{array}
\newcolumntype{C}{>{\centering\arraybackslash}X}
\newcolumntype{L}{>{\raggedright\arraybackslash}X}

\usepackage{amsfonts}
\usepackage{graphicx}
\usepackage{epstopdf}
\ifpdf
  \DeclareGraphicsExtensions{.eps,.pdf,.png,.jpg}
\else
  \DeclareGraphicsExtensions{.eps}
\fi

\usepackage[linesnumbered,ruled,vlined]{algorithm2e}
\makeatletter
\newcommand{\nosemic}{\renewcommand{\@endalgocfline}{\relax}}
\newcommand{\dosemic}{\renewcommand{\@endalgocfline}{\algocf@endline}}
\newcommand{\pushline}{\Indp}
\newcommand{\popline}{\Indm\dosemic}
\let\oldnl\nl
\newcommand{\nonl}{\renewcommand{\nl}{\let\nl\oldnl}}
\makeatother

\newcommand{\email}[1]{\href{mailto:#1}{#1}}

\def\R{{\mathbb{R}}}
\def\Z{{\mathbb{Z}}}
\def\F{{\mathcal{F}}}
\def\E{{\mathcal{E}}}

\graphicspath{{./Images}}

\title{Demons registration for 2D empirical wavelet transforms}

\author{Charles-G\'erard Lucas\thanks{Department of Mathematics \& Statistics, San Diego State University,
5500 Campanile Dr., San Diego, 92182, CA, USA.
  (\email{clucas2@sdsu.edu}, \email{jgilles@sdsu.edu}).
  }
\and J\'er\^ome Gilles\footnotemark[1]
}

\title{Demons registration for 2D empirical wavelet transforms}

\date{}

\begin{document}
\emergencystretch 3em

\maketitle

\begin{abstract}
The empirical wavelet transform is a fully adaptive time-scale representation that has been widely used in the last decade. Inspired by the empirical mode decomposition, it consists of filter banks based on harmonic mode supports. Recently, it has been generalized to build the filter banks from any generating function using mappings. In practice, the harmonic mode supports can have a low-constrained shape in 2D, leading to numerical difficulties to estimate mappings adapted to the construction of empirical wavelet filters. This work aims to propose an efficient numerical scheme to compute empirical wavelet coefficients using the demons registration algorithm. Results show that the proposed approach is robust, accurate, and continuous wavelet filters permitting reconstruction with a low signal-to-noise ratio. An application for texture segmentation of scanning tunneling microscope images is also presented.
\end{abstract}

{\small \textbf{Keywords:} Empirical wavelet, adaptive representation, demons algorithm, texture segmentation. }

\section{Introduction}
The empirical wavelet transform is a fully adaptive time-scale representation, introduced in~\cite{Gilles2013}, based on data-driven filter banks. Unlike the traditional wavelet transform, the~wavelet filters are not designed on the basis of prescribed scales independent of the data under study, but~on the basis of supports that contain the underlying harmonic modes. This approach is inspired by the empirical mode decomposition~\cite{huang1998empirical} but~is based on solid theoretical foundations that the latter lacks.
Due to its robust performance in decomposing images, it has led to many applications such as glaucoma detection~\cite{maheshwari2016automated,parashar2020automatic}, hyperspectral image classification~\cite{prabhakar2017two}, cancer histopathological image classification~\cite{deo2024ensemble}, medical image fusion~\cite{sundar2017multi,polinati2020multimodal} and texture segmentation~\cite{Huang2019}. Notably, it has been shown to outperform traditional wavelet transforms in extracting appropriate texture features from images~\cite{huang2018review,hurat2020empirical}.

The last decade has seen an intensive development of the two steps involved in the 2D empirical wavelet transform: $(i)$ the extraction of regions containing the harmonic modes and $(ii)$ the design of filters mainly supported by these regions.
The extraction of harmonic modes can be performed by scale-space representations~\cite{Gilles2014a}, and several methods of separation of the harmonic modes in possibly low-constrained shaped supports have been proposed, such as the Tensor~\cite{Gilles2013}, Ridgelet, Curvelet~\cite{hurat2020empirical}, Watershed~\cite{Gilles2013a} and Voronoi~\cite{gilles2022empirical} methods. However, the~design of filter banks has mainly long been limited to a specific generating function based on the Littlewood--Paley formulation. This formulation relies on a separable definition or the distance to the boundaries of the supports, which limits its extension to other generating functions. 

Recently, a general framework has been proposed to consider classic generating functions~\cite{lucas2024multidimensional}. 
Deformations of a homogeneous or separable wavelet kernel are carried out by mappings between the data's Fourier supports and the generating function's Fourier support. In~practice, this approach suffers from the difficulty of estimating the required mappings with constraints of invertibility, continuity, and differentiability.
Such an estimation can be performed by the well-known demons algorithm~\cite{thirion1998image}, as first proposed by~\cite{lucas2024multidimensional}, or its variants~\cite{cachier2003iconic,vercauteren2009diffeomorphic}. These fast algorithms have been widely used and compared for medical image registration~\cite{wang2018adaptive,tabassum2020registration}, where they have shown to provide robust and accurate estimates. However, an~in-depth investigation of their use for empirical wavelet transform remains a key challenge. This is the main objective addressed in this work.

This manuscript is organized as follows. In Section~\ref{sec:2Dewt}, we recall the framework of empirical wavelet transform based on mappings and present different demons algorithms for mapping estimation. In particular, as~a first contribution, we build specific band-pass wavelet filters from homogeneous or separable wavelet kernels, which is suitable for an empirical analysis of the wavelet transform.
Section~\ref{sec:exp} presents the main contribution of this manuscript: a comparative analysis of various demons algorithms for designing continuous wavelet filters. The~accuracy of the estimated filters and the effect of the estimation process on wavelet reconstruction is thoroughly analyzed. Additionally, to~demonstrate the practical utility of the proposed tools, we include an illustration of texture segmentation for scanning tunneling microscope images. Section~\ref{sec:exp} provides a detailed discussion of the~\mbox{findings}.  The~\textsc{MATLAB} toolbox for the empirical wavelet transform is freely available at \url{https://github.com/jegilles/Empirical-Wavelets}.

\section{Notations}
We recall that an invertible function $\gamma$ is called a homeomorphism if it is continuous of inverse continuous and a diffeomorphism if it is infinitely differentiable of inverse infinitely differentiable. 
We consider that the space of square integrable functions $\mathrm{L}^2(\R^2)$ is endowed with the usual inner product
$$\langle f,g \rangle =\int_{\R^2} f(x) \overline{g(x)} \mathrm{d} x.$$
The Fourier transform $\widehat{f}$ of a function $f \in \mathrm{L}^2(\R^2)$ and its inverse are given by, respectively,
\begin{align*}
& \widehat{f}(\xi) = \F (f)(\xi) = \int_{\R^2} f(x)e^{-2\pi i (\xi \cdot x)} \mathrm{d} x, \\
& f(x) = \F^{-1} (\widehat{f})(x) = \int_{\R^2} \widehat{f}(\xi)e^{2\pi i (\xi \cdot x)} \mathrm{d} \xi,
\end{align*}
where $\cdot$ stands for the dot product in $\R^2$. 

\section{Empirical wavelet transform}
\label{sec:2Dewt}

\subsection{Empirical wavelet systems}

The construction of empirical wavelets relies on the partitioning of the Fourier domain $\Omega$. Technically, we consider a family of disjoint open sets $\{\Omega_n\}_{n\in\Upsilon}$, with~$\Upsilon \subset \Z$, of~closures $\overline{\Omega}_n$ covering the Fourier domain, i.e.,~$\Omega=\bigcup_{n\in\Upsilon} \overline{\Omega_n}$. 
In this work, we focus on real-valued images, which implies that the Fourier domain is symmetric, leading us to consider a partition $\{\Omega_n\}_{n\in\Upsilon}$ that is symmetric, i.e.,~such that $\Omega_{-n} = \{-\xi \mid \xi \in \Omega_n \}$. To~obtain such a partition, the~modes of the Fourier spectrum are obtained by scale-space representation~\cite{Gilles2014a} , and the Fourier domain is partitioned by the Watershed~\cite{hurat2020empirical} or Voronoi~\cite{gilles2022empirical} methods. Figure~\ref{fig:images} shows an example of an image and the Watershed and Voronoi partitions of the logarithm of its Fourier spectrum using a scale-space step-size parameter set to $s_0=0.8$.

Symmetric empirical wavelet systems are built as filters mostly supported on the sets $\Omega_n \cup \Omega_{-n}$ for $n \in\Upsilon^+=\{n\in\Upsilon \mid n \geq 0\}$.
The authors of~\cite{lucas2024multidimensional}have  proposed building filters $\chi_n$ from a wavelet kernel $\psi$ using mappings $\gamma_n$ on $\R^2$. Let $\psi$ be such that its Fourier transform $\widehat{\psi}$ is localized in frequency and compactly supported (or rapidly decaying) by a connected open subset $\Lambda$.
Let $\{\gamma_n\}_{n\in\Upsilon}$ be mappings on $\R^2$ such that $\Lambda=\gamma_n(\Omega_n)  \textrm{ if } \Omega_n$ is bounded and $\Lambda \subsetneq \gamma_n(\Omega_n)$ otherwise. Empirical wavelet systems $\psi_n$ and their symmetric counterparts $\chi_n$ are then defined as~follows. 

\begin{figure}[!t]
\centerline{
  \includegraphics[width=.3\textwidth]{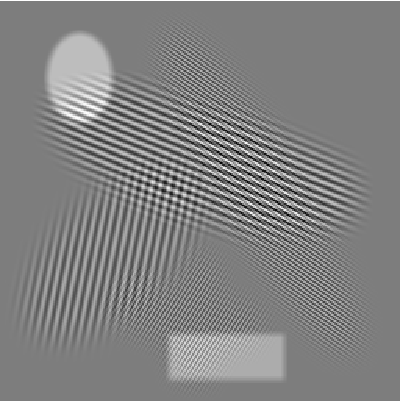}
  \includegraphics[width=.3\textwidth]{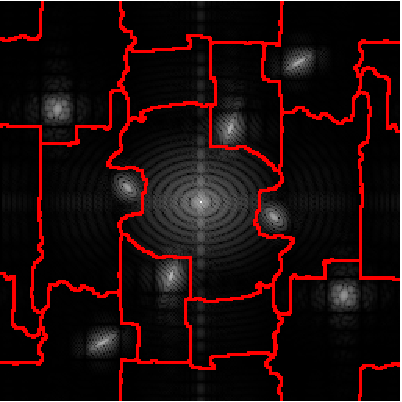}
  \includegraphics[width=.3\textwidth]{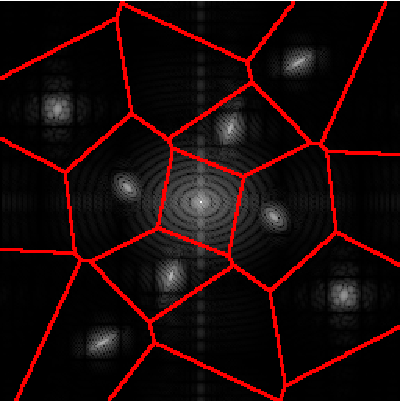}
   }
 \caption{{\bf Fourier partitioning.} Toy image of size $256 \times 256$ along with the (middle) Watershed and (right) Voronoi partitions (overlapping in red) of the logarithm of their Fourier spectra.}
  \label{fig:images}
\end{figure}

Symmetric empirical wavelet systems are built as filters mostly supported on the sets $\Omega_n \cup \Omega_{-n}$ for $n \in\Upsilon^+=\{n\in\Upsilon \mid n \geq 0\}$ the subset of positive elements of $\Upsilon$.
In \cite{lucas2024multidimensional}, it has been proposed to build filters $\chi_n$ from a wavelet kernel $\psi$ using mappings $\gamma_n$ on $\R^2$. Let $\psi$ such that its Fourier trasnform $\widehat{\psi}$ is localized in frequency and compactly (or rapidly decaying) supported by a connected open subset $\Lambda$.
Let $\gamma_n$ be mappings on $\R^2$ such that $\Lambda=\gamma_n(\Omega_n)  \textrm{ if } \Omega_n$ is bounded and $\Lambda \subsetneq \gamma_n(\Omega_n)$ otherwise. Empirical wavelet systems $\psi_n$ and their symmetric counterparts $\chi_n$ are then defined as follows. 

\begin{Definition}[Normalized empirical wavelet system]
\label{def:norm_system}
Assume that the mappings $\gamma_n$ are diffeomorphisms.
The symmetric empirical wavelet system $\{\chi_n\}_{n\in\Upsilon}$ corresponding to the partition $\{\Omega_n\}_{n\in\Upsilon}$ is defined by, for all $\xi \in\R^2$,
\begin{equation*}
\widehat{\chi}_0(\xi) = \widehat{\psi}_0(\xi) \quad  \textrm{ and } \quad \forall n\in\Upsilon\setminus\{0\}, \quad \widehat{\chi}_n(\xi) = \frac{1}{\sqrt{2}} \left( \widehat{\psi}_n(\xi) + \widehat{\psi}_{-n}(\xi) \right),
\end{equation*}
where, for every $n\in\Upsilon$,
\begin{equation*}
\widehat{\psi}_n(\xi)= \frac{1}{\sqrt{a_n(\xi)}} \; \widehat{\psi} \circ \gamma_n(\xi), \qquad \textrm{with } \quad a_n(\xi)=\frac{1}{\vert \mathrm{det}\; J_{\gamma_n}(\xi) \vert},
\end{equation*}
with $\gamma_{-n}:\xi \mapsto -\gamma_n(-\xi)$ and $J_{\gamma_n}$ being the Jacobian of the mapping $\gamma_n$.
\end{Definition}

The symmetry of the filter $\widehat{\chi}_n$ has to be understood with respect to the Fourier support $n$, i.e. $\widehat{\chi}_n=\widehat{\chi}_{-n}$.
The normalization coefficient $a_n(\xi)$ ensures that
\begin{equation}
\label{eq:substitution}
\int_{\Omega_n} \vert \widehat{\chi}_n (\xi) \vert^2 \mathrm{d} \xi = \int_{\Lambda} \vert \widehat{\psi}(\xi) \vert^2 \mathrm{d} \xi,
\end{equation}
in most cases (see \cite{lucas2024multidimensional} for more details).
If $\gamma_n$ is not a diffeomorphism, the substitution theorem ensuring the energy conservation \eqref{eq:substitution} is not valid despite the use of the normalization coefficient $a_n(\xi)$. 
However, in practice, estimating diffeomorphisms is a difficult and computationally expensive task.
Since an intervertible continuous function is a homeopmorphism if and only if it is an open map, i.e., a map for which the preimage of an open set is an open set, the existence of homeomorphisms is a milder assumption and their estimation is less difficult.
We thus propose an unnormalized definition for the case when $\gamma_n$ is only a homeomorphism by removing the Jacobian determinant from Definition~\ref{def:norm_system}, as follows.

\begin{Definition}[Unnormalized empirical wavelet system]
\label{def:unnorm_system}
Assume that the mapping $\gamma_n$ is a homeomorphism.
The unnormalized symmetric  empirical wavelet system $\{\chi_n\}_{n\in\Upsilon}$ corresponding to the partition $\{\Omega_n\}_{n\in\Upsilon}$ is defined by, for all $\xi \in\R^2$,
\begin{equation*}
\widehat{\chi}_0(\xi) = \widehat{\psi}_0(\xi) \quad  \textrm{ and } \quad \forall n\in\Upsilon\setminus\{0\}, \quad \widehat{\chi}_n(\xi) = \widehat{\psi}_n(\xi) + \widehat{\psi}_{-n}(\xi),
\end{equation*}
where $\widehat{\psi}_n= \widehat{\psi} \circ \gamma_n$ and $\gamma_{-n}:\xi \mapsto -\gamma_n(-\xi)$. 
\end{Definition}

\subsection{{Empirical Wavelet~Transform}}

The empirical wavelet transform consists of a filtering process from a normalized or unnormalized empirical wavelet system, as follows.

\begin{Definition}[Empirical wavelet transform]
The symmetric empirical wavelet transform is defined by, for all $n\in\Upsilon$, 
\begin{equation*}
\E_{\chi}^f(\cdot,n) = \F^{-1}\left(\widehat{f} . \overline{\widehat{\chi}_n}\right),
\end{equation*}
where $\F$ denotes the inverse Fourier transform. 
\end{Definition}

Figure~\ref{fig:ewt_filtering} shows an example of unnormalized wavelet filter's Fourier transforms $\widehat{\chi}_n$ and the resulting wavelet coefficients $\E_{\chi}^f(\cdot,n)$ for the toy image in Figure~\ref{fig:images} (left) and the Watershed partition of its Fourier spectrum given in Figure~\ref{fig:images} (middle).

\begin{figure}[!h]
\centerline{
  \includegraphics[width=\textwidth]{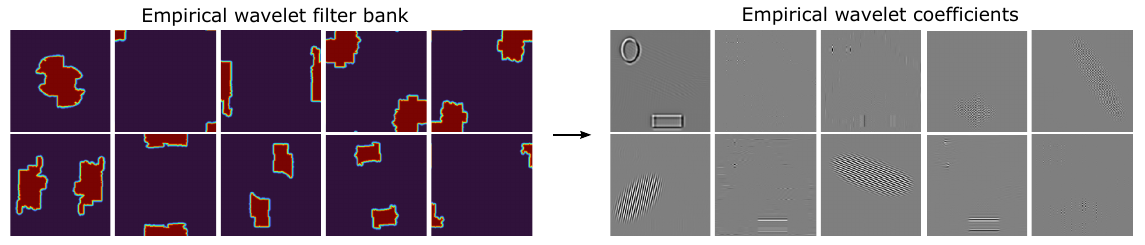}
   }
 \caption{{\bf Empirical wavelet filtering.} Example of unnormalized symmetric empirical wavelet filters $\chi_n$ in the Fourier domain (left) and resulting wavelet coefficients $\E_{\chi}^f(\cdot,n)$ (right), for $n=1,\ldots,10$, associated with the toy image using the Watershed partitioning.}
  \label{fig:ewt_filtering}
\end{figure}

The reconstruction of a real-valued image $f$ from its symmetric empirical wavelet transform $\E_{\chi}^f$ is guaranteed by Theorem 4.3 in \cite{lucas2024multidimensional}, recalled hereafter.

\begin{Theorem}[Reconstruction]
\label{th:reconstruction}
Let $\Upsilon^+ = \{n \in \Upsilon: n \geq 0 \}$. Let assume that, for $a.e.\, \xi \in \R^2$,
$0<\sum_{n\in\Upsilon^+}\left|\widehat{\psi}_n(\xi)\right|^2<\infty$ and $\sum_{n\in\Upsilon^+\setminus\{0\}}\left|\widehat{\psi}_n(\xi)\right| \left|\widehat{\psi}_{-n}(\xi)\right|<\infty.$
Then 
\begin{align}
\label{eq:symreconstruction}
f = \sum_{n\in\Upsilon^+} \E_{\chi}^f(\cdot,n)\star \F^{-1}\left(\frac{\widehat{\chi}_n}{\displaystyle\sum_{m\in\Upsilon^+}\left|\widehat{\chi}_m\right|^2} \right),
\end{align}
where $\star$ denotes the convolution of functions and $\F^{-1}$ denotes the inverse Fourier transform.
\end{Theorem}

\subsection{Band-pass empirical wavelets}

For the study of the numerical behavior of the empirical wavelet transform, it is more convenient to consider band-pass filters. Therefore, we propose in this section defining homogeneous or separable 2D wavelet band-pass filters.
The proposed definition extends the band-pass wavelet filters introduced in~\cite{hurat2020empirical, Gilles2013}, which are based on the semi-Euclidean distance to support boundaries, thus corresponding to a discrete mapping. In 1D, the~continuous band-pass wavelet kernel can be defined as follows.
\begin{Definition}[1D band-pass wavelet]
The 1D band-pass wavelet, mostly supported by the segment $\Lambda=[-\frac{1}{2},\frac{1}{2}]$, is defined by, for every $\xi \in \R$,
\begin{equation*}
\widehat{\psi}_{\tau}^{\rm 1D}(\xi) =  
\begin{cases}
1 & \textrm{if } \vert \xi \vert < \frac{1}{2} - \tau, \\
\displaystyle \cos \frac{\pi}{2} \beta \left(\frac{\tau - \frac{1}{2} + \vert \xi \vert}{2 \tau} \right) & \textrm{if }  \frac{1}{2} - \tau \leq \vert \xi \vert \leq \frac{1}{2} + \tau, \\
0 & \textrm{otherwise},
\end{cases}
\end{equation*}
where $0 < \tau < \frac{1}{2}$ is a transition width and $\beta$ is a continuous function on $[0,1]$ such that $\beta(0)=0$, $\beta(1)=1$ and $\beta(x)+\beta(1-x)=1$ for every $x \in [0,1]$.
\end{Definition}
The function $\beta$ is usually chosen as $\beta(x) = x^4(35-84x+70x^2-20x^3)$. 
The 1D band-pass wavelet can be extended to 2D in a homogeneous or separable way, as proposed in the following definitions.

\begin{Definition}[Disk band-pass wavelet]
The disk band-pass wavelet, mostly supported by the disk $\Lambda = \mathrm{B}(0,\frac{1}{2})$, is defined by, for every $\xi \in \R^2$,
\begin{equation*}
\widehat{\psi}^{\mathrm{D}}_\tau(\xi) =  
\begin{cases}
1 & \textrm{if } \Vert \xi \Vert_2 < \frac{1}{2} - \tau, \\
\displaystyle \cos \frac{\pi}{2} \beta \left(\frac{\tau - \frac{1}{2} + \Vert \xi \Vert_2}{2 \tau} \right) 
& \textrm{if }  \frac{1}{2} - \tau \leq \Vert \xi \Vert_2 \leq \frac{1}{2} + \tau, \\
0 & \textrm{otherwise},
\end{cases}
\end{equation*}
where $0 < \tau < \frac{1}{2}$ is a transition width and $\beta$ is a continuous function on $[0,1]$ such that $\beta(0)=0$, $\beta(1)=1$ and $\beta(x)+\beta(1-x)=1$ for every $x \in [0,1]$.
\end{Definition}

\begin{figure}[!t]
\centerline{
  \includegraphics[width=.25\textwidth]{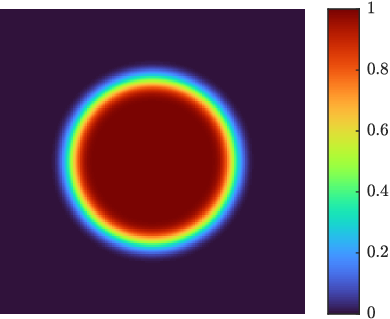}
  \hspace{.1\textwidth}
  \includegraphics[width=.25\textwidth]{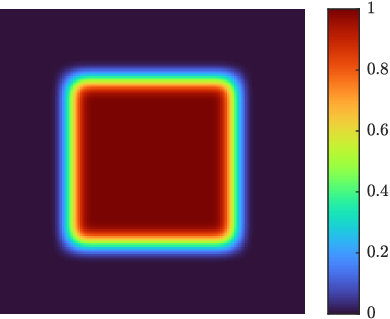}
   }
 \caption{{\bf Band-pass wavelet.} Disk (left) and square (right) band-pass wavelet with transition width $\tau=0.2$.}
  \label{fig:bandpass_wavelet}
\end{figure}

\begin{Definition}[Square band-pass wavelet]
The square band-pass wavelet, mostly supported by the square $\Lambda=[-\frac{1}{2},\frac{1}{2}]^2$, is defined by, for every $\xi=(\xi_1,\xi_2) \in \R^2$,
\begin{equation*}
\widehat{\psi}^{\mathrm{S}}_\tau(\xi) = \widehat{\psi}_{\tau}^{\rm 1D}(\xi_1) \; \widehat{\psi}_{\tau}^{\rm 1D}(\xi_2),
\end{equation*}
where $0 < \tau < \frac{1}{2}$ is a transition width.
\end{Definition}

Figure~\ref{fig:bandpass_wavelet} show examples of disk and square band-pass wavelet filters.

\begin{Theorem}
\label{th:disk_frames}
Let $\chi^{\mathrm{D},\tau}_n = \chi^{\mathrm{D}}_\tau \circ \gamma_n$ and $\chi^{\mathrm{S},\tau}_n = \chi^{\mathrm{S}}_\tau \circ \gamma_n$, the (normalized or unnormalized) symmetric wavelet systems resulting from the disk and square band-pass wavelets $\widehat{\psi}_{\tau}^D$ and $\widehat{\psi}_{\tau}^S$, respectively, for~every $n \in \Upsilon$. The~reconstruction Formula \eqref{eq:symreconstruction} holds for the normalized and unnormalized systems $\{\chi^{\mathrm{D},\tau}_n\}_{n \in \Upsilon}$ and $\{\chi^{\mathrm{S},\tau}_n\}_{n \in \Upsilon}$.
\end{Theorem}
\begin{proof}
Let  $\widehat{\psi}_{\tau}$ be either $\widehat{\psi}_{\tau}=\widehat{\psi}_{\tau}^D$ or $\widehat{\psi}_{\tau}=\widehat{\psi}_{\tau}^S$. We set $\Lambda_{\tau} = \mathrm{B}(0,\frac{1}{2}+\tau)$ if $\widehat{\psi}_{\tau}=\widehat{\psi}_{\tau}^D$ and 
$\Lambda_{\tau} =[-\frac{1}{2}-\tau,\frac{1}{2}+\tau]^2 $ if $\widehat{\psi}_{\tau}=\widehat{\psi}_{\tau}^S$. Let $\xi \in \Omega_m$. 
Since there exists a finite number of $\gamma_n$ such that $\gamma_n(\xi) \in \Lambda_\tau$, we have
\begin{equation}
\sum_{n \in \Upsilon^+} \left\vert \widehat{\psi}_n(\xi) \right\vert^2 = \sum_{n \in \Upsilon^+} \left\vert \frac{1}{a_n(\xi)} \right\vert \left\vert \widehat{\psi}_{\tau}(\gamma_n(\xi)) \right\vert^2 =  \sum_{n \in \Upsilon_m^+} \left\vert \frac{1}{a_n(\xi)} \right\vert \left\vert \widehat{\psi}_{\tau}(\gamma_n(\xi)) \right\vert^2 < +\infty,
\end{equation}
where $\Upsilon_m^+ = \{n \in \Upsilon^+:\gamma_n(\xi) \in \Lambda_\tau\}$ and $a_n(\xi)=1$ or $a_n(\xi)=1/\vert \mathrm{det}\; J_{\gamma_n}(\xi) \vert$ for the normalized or unnormalized empirical wavelet system, respectively. 
Moreover, since $\widehat{\psi}_{\tau}(\gamma_m(\xi))\neq 0$ and $a_n(\xi) > 0$ for every $n \in \Upsilon$, we have $0 < \sum_{n \in \Upsilon^+} \vert \widehat{\psi}_n(\xi) \vert^2$. 
Similarly, we can show that $\sum_{n \in \Upsilon^+} \vert \widehat{\psi}_n(\xi) \widehat{\psi}_{-n}(\xi) \vert < +\infty$. Hence, the reconstruction Formula \eqref{eq:symreconstruction} for $\{\chi^{\mathrm{D},\tau}_n\}_{n \in \Upsilon}$ and $\{\chi^{\mathrm{S},\tau}_n\}_{n \in \Upsilon}$ holds by Theorem~\ref{th:reconstruction}.
\end{proof}

\subsection{Mapping estimation}
In practice, the~mappings $\gamma_n$ in Definitions \ref{def:norm_system} and \ref{def:unnorm_system} have to be estimated.
The Thirion's demons algorithm~\cite{thirion1998image} is a mapping estimation scheme inspired by a diffusion process in which the targeted mapping is represented by a displacement field $\gamma$.
This method alternates between solving the flow equations and regularization. To~give a theoretical justification for this method,~\cite{vercauteren2009diffeomorphic} revisited it as the following optimization problem:
\begin{equation}
\label{eq:demons_opti}
\underset{\gamma,c \in \F(\R^2)}{\mathrm{minimize}} \quad \Vert \Lambda - \Omega \circ c \Vert_2^2 + \frac{1}{\sigma_x^2} \Vert \gamma - c \Vert_2^2 + \mathcal{R}(\gamma),
\end{equation}
where $\Lambda$ and $\Omega$ are considered as functions from $\R^2$ to $\R$, $\F(\R^2)$ denotes the set of mappings on $\R^2$, $c$ is an intermediate mapping to allow spatial uncertainties, $\sigma_x$ is a parameter controlling the spatial uncertainty on the mappings, and $\mathcal{R}(\gamma)$ is a regularization term. 

The additive demons algorithm proposed in~\cite{cachier2003iconic} consists of estimating the minimum $s$ of Equation~\ref{eq:demons_opti} using an intermediate displacement field $u$ such that $c = \gamma + u$ and~alternate
\begin{align}
& \displaystyle \underset{u \in \F(\R^2)}{\mathrm{minimize}}
\quad \Vert \Lambda - \Omega \circ (\gamma + u) \Vert_2^2 + \frac{1}{\sigma_x^2}\Vert u \Vert_2^2 := E(u), \label{eq:u_minimization} \\
& \displaystyle \underset{\gamma \in \F(\R^2)}{\mathrm{minimize}}
\quad \frac{1}{\sigma_x^2} \Vert \gamma - c \Vert_2^2 + \mathcal{R}(\gamma). \label{eq:s_minimization}
\end{align}
The displacement field $u$ minimizing Equation~\ref{eq:u_minimization} is approximated by, for~every position $p \in R^N$,
\begin{equation}
\label{eq:demons_reg_update}
\displaystyle u(p) \approx \frac{\Lambda(p) - \Omega \circ \gamma(p)}{ \displaystyle \Vert \nabla (\Omega \circ \gamma)(p) \Vert_2^2 - \frac{\sigma_i^2}{\sigma_x^2} \vert \Lambda(p) - \Omega \circ \gamma(p) \vert^2 } \nabla \Omega^\top,
\end{equation}
where $\nabla$ denotes the gradient of functions.
This implies that $\sigma_x$ controls the maximum step length: $\Vert u(p) \Vert_2 \leq \sigma_x/2$ for every $p \in \R^2$.
As for the minimization of Equation~\ref{eq:s_minimization} with a regularization term $\mathcal{R}(\gamma) = \Vert J_\gamma \Vert_2^2/\sigma_d^2$, it corresponds to a convolution with a Gaussian kernel $G_{\rm diff}(\sigma_d)$ of standard deviation $\sigma_d$. This regularization $\mathcal{R}(\gamma)$ has been modified by~\cite{cachier2003iconic} in order to add a fluid-like regularization $G_{\rm fluid}(\sigma_f)$ of standard deviation $\sigma_f$.
Moreover, to~ensure the mapping $\gamma$ to be diffeomorphic, an~alternative update of $c$ with an exponential field has been proposed in~\cite{vercauteren2009diffeomorphic}: $c = s \circ \exp(u)$. The~additive and diffeomorphic demons algorithms are summarized in Algotirhm~\ref{alg:demons}. 

\begin{algorithm}[!ht]
    \caption{Vercauteren's demons algorithm}
    \label{alg:demons}
    \DontPrintSemicolon
    \textbf{Input:} Sets $\Lambda$, $\Omega$. \\
    \textbf{Initialization:} Displacement field $\gamma^{[0]}$.\\
    $k = 0$~; \\
    \While  {$k \leq K$ and $\vert E(\gamma^{[k]})-E(\gamma^{[k-5]}) \vert / E(\gamma^{[0]}) > \epsilon$} {
    $k \leftarrow k+1$~; \\
        Update $u^{[k]}$ from $\gamma^{[k-1]}$ using Equation~\ref{eq:demons_reg_update}~; \\
        Fluid-like regularization: $\gamma^{[k]} = G_{\rm fluid}(\sigma_f) \star u^{[k]}$~; \\
        \nosemic Either additive update $c^{[k]} = \gamma^{[k-1]} + u^{[k]}$ \\
        \pushline\dosemic\nonl or diffeomorphic update $c^{[k]} = \gamma^{[k-1]} \circ \exp(u^{[k]})$ ~; \\
    	\popline Diffusion-like regularization: $\gamma^{[k]} = G_{\rm diff}(\sigma_d) \star c^{[k]}$~;
    } 
        \textbf{Output:} Displacement field $\gamma^{[k]}$ minimizing $E(\gamma^{[k]})$ over $k$.
    
\end{algorithm}

For large deformations, a multiresolution scheme is necessary: the demons is performed iteratively from the lowest to the highest resolution. This scheme is summarized in Algotirhm~\ref{alg:multiresolution}.

\begin{algorithm}[!ht]
    \caption{Multiresolution demons algorithm}
    \label{alg:multiresolution}
    \DontPrintSemicolon
    \textbf{Input:} Sets $\Lambda$, $\Omega$. \\
    \textbf{Initialization:} Displacement field $\gamma^{[0]} = \mathbf{0}$.\\
    \For  {$k=N_{\rm level}-1,\ldots,0$} {
    	 $\Lambda^{[k]}, \Omega^{[k]}, \gamma^{[k]}$ $\leftarrow$ Downsample $\Lambda$, $\Omega$ and $\gamma^{[k-1]}/2^{k}$ by $2^{k}$ \\
    	$\gamma^* \leftarrow \textrm{Demons}(\Lambda^{[k]},\Omega^{[k]})$ initialized with $\gamma^{[k]}$   $\quad$ (Algotirhm~\ref{alg:demons})\\
    	$\gamma^{[k]}$ $\leftarrow$ Upsample $2^{k} \gamma^*$ by $2^{k}$
    	}
        \textbf{Output:} Displacement field $\gamma^{[N_{\rm level}]}$ \\
\end{algorithm}

\section{Numerical experiments} 
\label{sec:exp}

In this section, we compare the Thirion's demons algorithm (\textsc{MATLAB} built-in function \emph{imregdemons}) and the additive and diffeomorphic Vercauteren's demons algorithms (\textsc{MATLAB} toolbox available at \url{https://www.mathworks.com/matlabcentral/fileexchange/39194-diffeomorphic-log-demons-image-registration}) for the computation of estimates $\widetilde{\gamma}_n$ of the mappings $\gamma_n$. We consider the toy image and the sets $\Omega_n$ obtained by either Watershed or Voronoi partitioning presented in Figure~\ref{fig:images}.

\subsection{Mapping estimation set-up}

For the three demons algorithms, the~diffusion-like parameter $\sigma_d$ and number of resolution levels $N_{\rm level}$ are selected by minimizing the quadratic risk $\Vert \Lambda - \Omega_n \circ \widetilde{\gamma}_n \Vert_2^2$ over the values $(\sigma_d,N_{level}) \in (0.3,0.31,\ldots,0.5)\times (n_P-1,n_P)$, with~$n_P$ being the highest integer such that $2^{n_P}$ is smaller than each dimension of the image.
For the Thirions's demons algorithm, the~numbers   iterations at each resolution level is set to $(2^4,\ldots,2^{N_{level}+1})$ from the highest to the lowest level. For~the Vercauteren's demons algorithm, the~maximum step length is set to $\sigma_x = 5$, the~fluid-like regularization is set to $\sigma_f=1$, the~error threshold in the stop criterion is set to $\epsilon=10^{-3}$ and the maximum number of iterations is set to $K=500$.

\subsection{Assessment measures}
The accuracy of an estimate $\widetilde{\gamma}_n$ is measured using the Root Mean Squared Error defined as
\begin{equation*}
\mathrm{RMSE}(\Lambda,\Omega_n,\widetilde{\gamma}_n) = \frac{1}{\sqrt{N_{\rm pixels}}} \Vert \Lambda - \Omega_n \circ \widetilde{\gamma}_n \Vert_2, 
\end{equation*}
where $N_{\rm pixels}$ is the number of pixels in the image $\Lambda$.
The reconstruction $\widetilde{f}$ of the image $f$, with values in $[0,1]$, obtained by Equation~\ref{eq:symreconstruction} is assessed using the Peak Signal-to-Noise Ratio (PSNR) defined as
\begin{equation*}
\mathrm{PSNR}(f,\widetilde{f}) =  - 10 \log_{10} \Vert \widetilde{f} - f \Vert_2^2.
\end{equation*}

\subsection{Mapping estimation assessment}

We assess the behavior of the different algorithms to map a partition to a generating function's Fourier support, which is usually a disk or a square. Table~\ref{tab:comp_diffeo_estim} reports the average RMSE of the estimated mappings $\widetilde{\gamma}_n$ from the sets $\Omega_n$ of the Watershed or Voronoi partition to the disk or square $\Lambda$ for the different demons algorithms. The~additive and Thirion's demons algorithms have similar performance {on average},
and both outperform the diffeomorphic demons algorithm due to the latter's additional constraint on diffeomorphicity. However, the~Thirion's demons algorithm suffers from more variability than the additive demons algorithm.
Furthermore, all demons algorithms achieve higher performance for the more shape-constrained Voronoi partition than for the Watershed partition  and~better for the disk than for the square due to its~irregularity.

\begin{table}[!h]
\centering
\caption{RMSE $\times 10^{-2}$ averaged over the sets $\Omega_n$ (with lowest and highest values) for the different demons algorithms, partitions $\{\Omega_n\}_{n\in\Upsilon}$ and sets $\Lambda$.}
\label{tab:comp_diffeo_estim}
\resizebox{\linewidth}{!}{
\begin{tabularx}{1.17\linewidth}{l llll}
  	\toprule
  \textbf{Demons} & \textbf{Watershed $\boldsymbol{\rightarrow}$ Disk} & \textbf{Watershed $\boldsymbol{\rightarrow}$ Square} & \textbf{Voronoi $\boldsymbol{\rightarrow}$ Disk} & \textbf{Voronoi $\boldsymbol{\rightarrow}$ Square}  \\
  \midrule
  Thirion's &  $4.37 \quad (0.82 - 18.43)$ & $4.59 \quad 
  (0.31 - 20.03)$  & $2.41 \quad (0.58 - 13.71)$ & $2.32 \quad (0.32 - 6.80)$ \\  
    Additive & $3.83 \quad (2.13 - 5.13)$ & $4.30 \quad (2.47 - 7.23)$  & $3.00 \quad (2.14 - 4.17)$ & $3.08 \quad (2.23 - 4.14)$ \\  
  Diffeomorphic  & $5.64 \quad (3.00 - 10.29)$ & $7.35 \quad (3.15 - 14.47)$ & $4.03 \quad (2.86 - 6.04)$ & $6.41 \quad (2.00 - 17.05)$ \\
  \bottomrule
\end{tabularx}
}
\end{table}

In addition, Table~\ref{tab:comp_diffeo_estim_time} reports the computtional time of the different demons algorithms for the Watershed and Voronoi partitions $\{\Omega_n\}_{n\in\Upsilon}$ and the disk and square sets $\Lambda$. The Thirion's demons is the fastest algorithm and the diffeomorphic demons is much more computationally costly than the two other demons.

\begin{table}[!h]
\centering
\caption{Computational times ($h$:$min$:$sec$) of the demons algorithms for the different partitions $\{\Omega_n\}_{n\in\Upsilon}$ and sets $\Lambda$.}
\label{tab:comp_diffeo_estim_time}
\resizebox{\linewidth}{!}{
\begin{tabularx}{1.16\linewidth}{lcccc}
  	\toprule
  \textbf{Demons} & \textbf{Watershed $\boldsymbol{\rightarrow}$ Disk} & \textbf{Watershed $\boldsymbol{\rightarrow}$ Square} & \textbf{Voronoi $\boldsymbol{\rightarrow}$ Disk} & \textbf{Voronoi $\boldsymbol{\rightarrow}$ Square}  \\
  \midrule
  Thirion's &  $00:02:44$ & $00:02:50$  & $00:02:48$ & $00:02:44$ \\  
  Additive & $00:12:50$ & $00:11:35$  & $00:14:51$ & $00:12:17$ \\  
  Diffeomorphic  & $01:43:22$ & $01:32:25$ & $01:27:53$ & $01:44:40$ \\
  \bottomrule
\end{tabularx}
}
\end{table}

To assess the invertibility, continuity, and diffeomorphicity of the mapping estimates, we explore the behavior of the resulting empirical wavelet coefficients expected to be concave and equal to one on most of each support $\Omega_n$. Figures  \ref{fig:normalized_wavelet_filters} and
\ref{fig:unnormalized_wavelet_filters} show the disk band-pass empirical wavelet coefficients $\widehat{\chi}_n^{\mathrm{D},\tau}$, where $\tau=0.2$, respectively, with and without normalization for the Watershed partition and the different demons algorithms. The~normalized coefficients are not concave for the additive and Thirion's demons algorithms due to the discontinuity of the normalization coefficient, while they are concave but very sparse for the diffeomorphic demons algorithm. The~unnormalized wavelet coefficients are more satisfactory as they are continuous, concave and mostly equal to one for all sets $\Omega_n$. However, the~coefficients obtained using the diffeomorphic demons algorithm are more spread out due to the additional constraint of differentiability. As~for the Thirion's demons algorithm, it can miss an important part of a filter (see the top right corner of Figure~\ref{fig:unnormalized_wavelet_filters}).

\begin{figure}[!h] 
\centering
\subfloat[Thirion's demons]{\includegraphics[width=.8\textwidth]{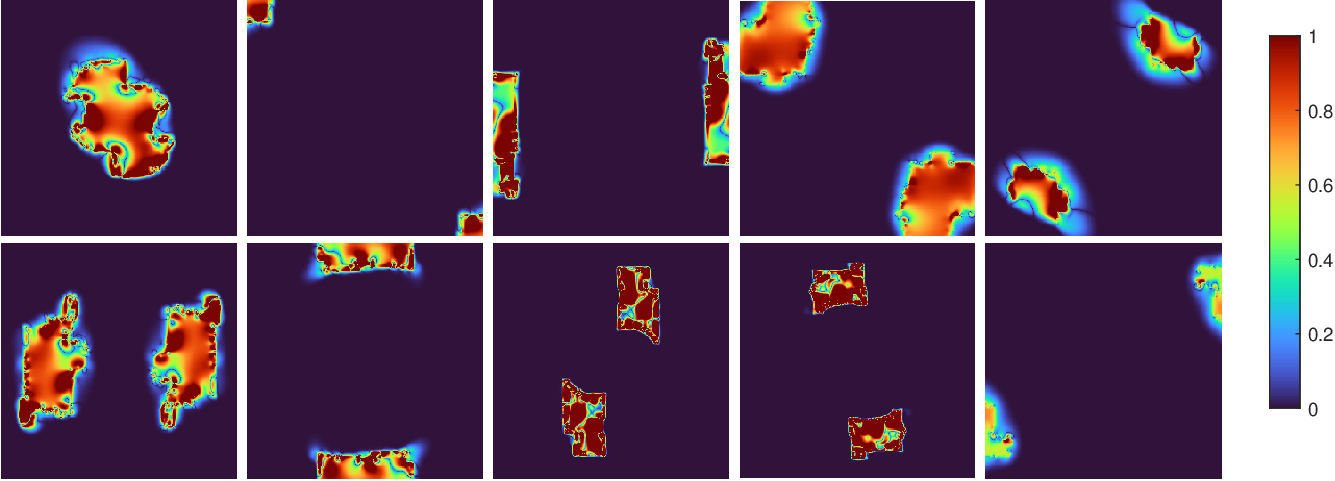} } \\

\subfloat[Additive demons]{\includegraphics[width=.8\textwidth]{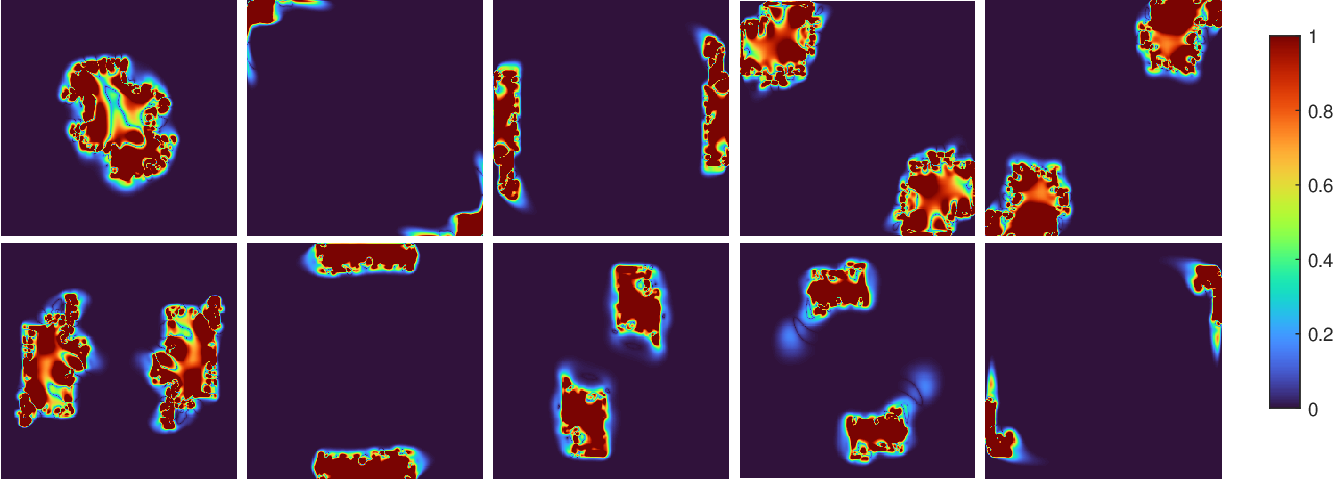} } \\

\subfloat[Diffeomorphic demons]{\includegraphics[width=.8\textwidth]{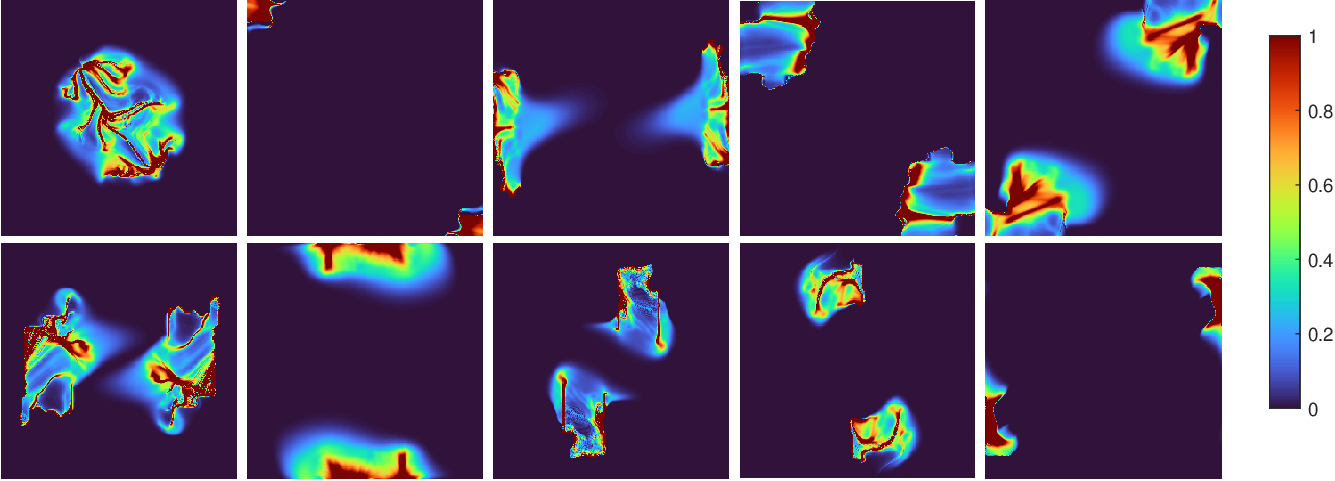} }
\caption{{\bf Normalized wavelet systems.} Normalized disk band-pass empirical wavelet filters $\widehat{\chi}_n^{\mathrm{D},\tau}$ with $\tau=0.2$ for the Watershed partition and the different mapping estimators.}
  \label{fig:normalized_wavelet_filters}
\end{figure}

\begin{figure}[!h] 
\centering
\subfloat[Thirion's demons]{\includegraphics[width=.8\textwidth]{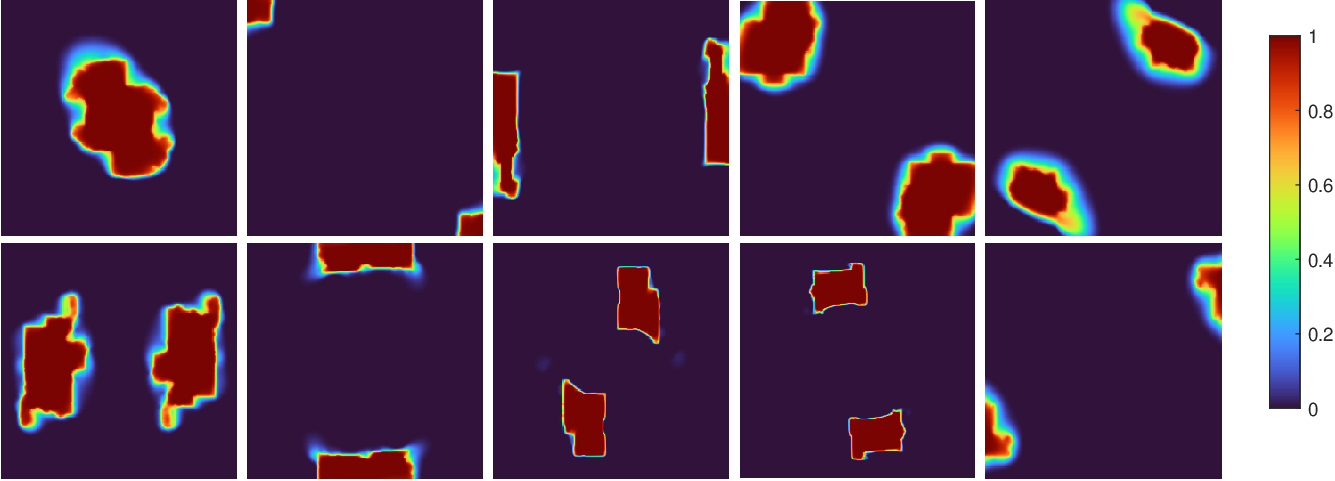}}  \\

\subfloat[Additive demons]{\includegraphics[width=.8\textwidth]{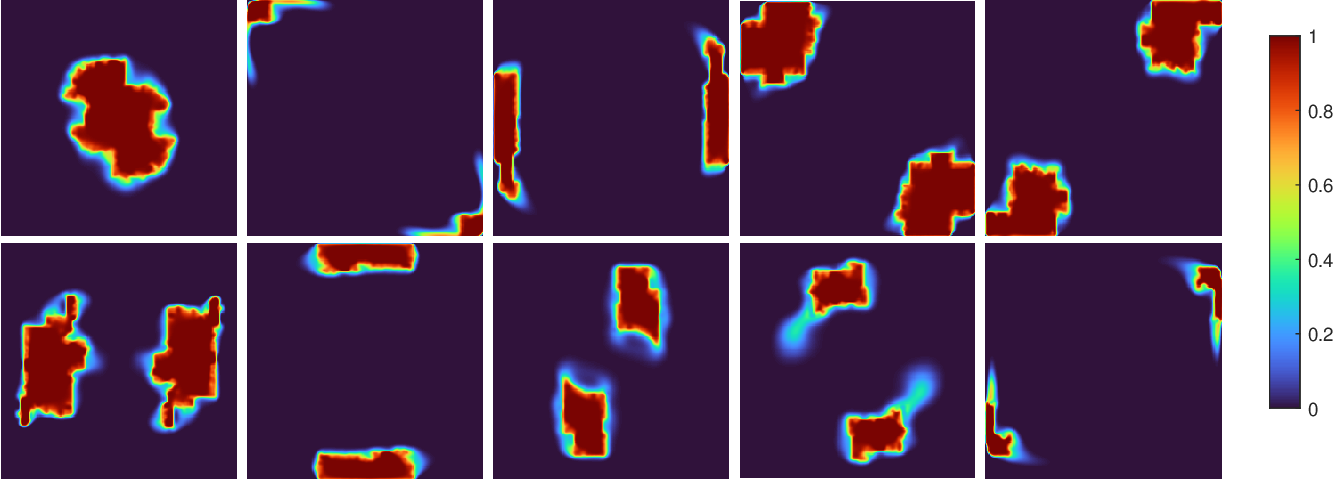}} \\

\subfloat[Diffeomorphic demons]{\includegraphics[width=.8\textwidth]{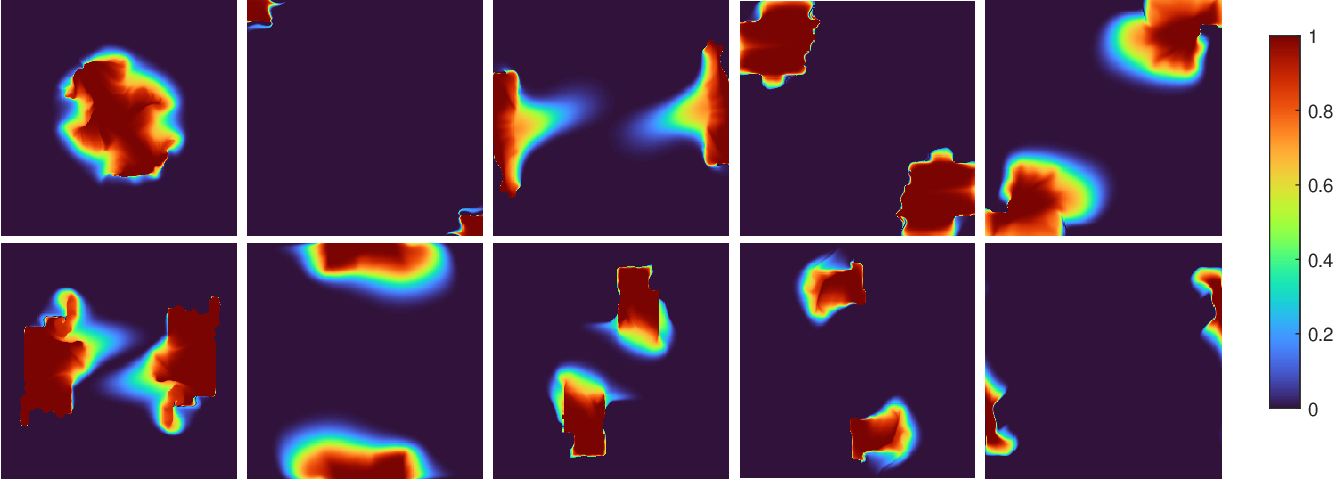}}
\caption{{\bf Unnormalized wavelet systems.} Unnormalized disk band-pass empirical wavelet filters $\widehat{\chi}_n^{\mathrm{D},\tau}$ with $\tau=0.2$ for the Watershed partition and the different mapping estimators.}
  \label{fig:unnormalized_wavelet_filters}
\end{figure}

In conclusion, the~additive demons algorithm provides accurate continuous mapping estimates to compute robust unnormalized empirical wavelet filters in a reasonable computation time. As~for the diffeomorphic demons algorithm, it provides inaccurate diffeomorphic mapping estimates, leading to unsuitable normalized~filters.

\subsection{Reconstruction assessment}

To assess the numerical behavior of the empirical wavelet transform, which is theoretically lossless for the disk and square band-pass filter according to Theorem~\ref{th:disk_frames}, we examine the quality of reconstruction. Table~\ref{tab:comp_normalized} reports the PSNR of the empirical wavelet reconstruction of the toy image for the disk and square band-pass filters with different transition widths $\tau$, the~Watershed and Voronoi partitions, and~the different demons algorithms. The~additive demons algorithm provides better, or~at least similar, reconstruction to the other demons algorithms. Moreover, the~normalization in the empirical wavelet filters has no impact on the reconstruction performance for the different demons~algorithms.

Despite the accuracy of the mapping estimation obtained in Table~\ref{tab:comp_diffeo_estim}, the~related wavelet reconstruction can be corrupted when the estimated filter bank does not entirely cover the Fourier domain. 
Notably, the~transition width $\tau=0.1$ leads to lower signal-to-noise ratios than the transition width $\tau=0.2$. Indeed, in~case of an inaccurate estimate of a mapping, $\tau=0.2$ ensures a larger covering of the Fourier domain by the wavelet filters than $\tau=0.1$.
In contrast, the~transition width $\tau=0.3$ can lead to artifacts in a symmetric wavelet filter $\chi_n$ due to overlaps of paired wavelet filters $\widehat{\psi}_n$ and $\widehat{\psi}_{-n}$, as confirmed by lower signal-to-noise ratios for the disk band-pass transform using the Watershed partition and a diffeomorphic demons algorithm compared to the transition width $\tau=0.2$.

\begin{table}[!ht]
\centering
\caption{PSNR of the reconstructed images for the different demons algorithms, sets $\Lambda$, partitioning methods and transition widths $\tau$ for the normalized and unnormalized empirical wavelet transforms.}
\label{tab:comp_normalized}
\newcolumntype{R}[1]{>{\let\newline\\\arraybackslash\hspace{0pt}}m{#1}}
\newcolumntype{C}[1]{>{\centering\let\newline\\\arraybackslash\hspace{0pt}}m{#1}}
\resizebox{\linewidth}{!}{
\begin{tabularx}{1.22\linewidth}{R{2.5cm} C{.7cm} C{1.4cm}C{1.4cm}C{1.4cm}C{1.4cm} C{1.4cm}C{1.4cm}C{1.4cm}C{1.4cm}}
  \toprule	
  \multirow{3}{*}{\textbf{Demons}} & \multirow{3}{*}{$\boldsymbol{\tau}$} &  \multicolumn{4}{c}{\textbf{Normalized}} & \multicolumn{4}{c}{\textbf{Unnormalized}} \\	
  & & \multicolumn{2}{c}{\textbf{Watershed}} & \multicolumn{2}{c}{\textbf{Voronoi}} & \multicolumn{2}{c}{\textbf{Watershed}} & \multicolumn{2}{c}{\textbf{Voronoi}} \\
  & & \textbf{Disk} & \textbf{Square} & \textbf{Disk} & \textbf{Square} & \textbf{Disk} & \textbf{Square} & \textbf{Disk} & \textbf{Square}\\
  \midrule
  Thirion's & \multirow{3}{*}{$0.1$} & $76.21$ & $\mathbf{79.67}$ & $82.61$ & $90.79$ & $76.21$ & $\mathbf{79.67}$ & $82.63$ & $90.79$  \\  
  Additive & & $\mathbf{96.08}$ & $79.23$ & $\mathbf{120.47}$ & $\mathbf{305.64}$ & $\mathbf{96.08}$ & $79.23$ & $\mathbf{120.47}$ & $\mathbf{306.03}$ \\
  Diffeomorphic & & $89.20$ & $62.05$ & $97.88$ & $286.93$  & $89.20$ & $62.05$ & $97.12$ & $269.60$ \\
  \midrule
  Thirion's & \multirow{3}{*}{$0.2$} & $78.22$ & $87.21$ & $85.34$ & $297.90$  & $78.22$ & $87.22$ & $85.34$ & $302.70$  \\ 
  Additive & & $\mathbf{112.78}$ & $\mathbf{305.65}$ & $\mathbf{305.47}$ & $\mathbf{305.69}$  & $\mathbf{112.78}$ & $\mathbf{305.88}$ & $\mathbf{306.37}$ & $\mathbf{306.26}$ \\
  Diffeomorphic & & $89.20$ & $124.85$ & $103.42$ & $303.76$  & $89.20$ & $120.77$ & $101.77$ & $305.40$ \\
  \midrule
  Thirion's & \multirow{3}{*}{$0.3$} & $85.02$ & $180.79$ & $89.32$ & $305.10$  & $84.82$ & $183.43$ & $89.32$ & $305.39$  \\ 
  Additive & & $\mathbf{124.41}$ & $\mathbf{305.16}$ & $\mathbf{305.78}$ & $\mathbf{305.61}$  & $\mathbf{124.41}$ & $\mathbf{305.67}$ & $\mathbf{306.00}$ & $\mathbf{306.16}$ \\
  Diffeomorphic & & $77.27$ & $117.87$ & $103.45$ & $303.19$  & $64.14$ & $128.81$ & $103.38$ & $305.38$ \\
  \bottomrule
\end{tabularx}
}
\end{table}

\subsection{Application to texture segmentation}

In this section, we illustrate the relevance of the proposed Empirical wavelet transform in the
segmentation of scanning tunneling microscope
images. Analyzing the structure and chemical properties of self-assembled monolayers of organic molecules is the core of many applications in nanoscience and nanotechnology~\cite{love2005self,gooding2003self}. These properties result in variations in the textures that are crucial to identify.
Recently, Empirical Curvelet Wavelet Transform (EWT-Curvelet), obtained by partitioning the Fourier domain in scales and angles~\cite{Gilles2013a}, has been shown to provide relevant texture features for scanning tunneling microscope images of self-assembled monolayers~\cite{guttentag2016hexagons,bui2020segmentation,guttentag2016surface}. 
In particular, the~EWT-Curvelet based segmentation outperforms the state-of-the-art methods for this task~\cite{bui2020segmentation}. Therefore, we propose a comparison of our approach with the latter.

We perform the texture segmentation according to the procedure of~\cite{huang2018review}, which consists of the following steps. 
First, we extract the cartoon and textural parts $u$ and $v$, respectively, of~an image $f$ by solving the cartoon–texture decomposition model~\cite{aujol2006combining}, defined as follows:
\begin{equation}
\displaystyle \underset{u \in F, v \in G}{\mathrm{minimize}} \quad \Vert \nabla u \Vert_1 + \mu \Vert v \Vert_G \quad \textrm{s.t.}  \;  f=u+v,
\end{equation}
where $F$ is the set of functions from $\R^2$ to $\R$ and $G$ is the set of oscillating functions. Following numerical analyses conducted by~\cite{gilles2012multiscale}, the~parameter controlling the oscillation degree of the textural part is set to $\mu=N_{\rm pixels}/2$. This problem is solved numerically using the algorithm described in~\cite{gilles2011bregman}. 
Next, we compute the local energy of the empirical wavelet transform $\E_{\chi}^v$ of the textural part $v$, defined as the local average of the wavelet coefficients $\E_{\chi}^v(.,n)$ for each frequency band $n$, i.e.,
\begin{equation}
\label{eq:local_energy}
\mathcal{T}^{\E}_n(x) = \frac{1}{N_W} \sum_{y\in W_x} \left\vert \E_{\chi}^v(y,n) \right\vert^2,
\end{equation}
where $W_x$ is a window centered at the pixel $x$ of size set to $N_W=19 \times 19$, following the work of~\cite{huang2018review}. Finally, the~segmentation is obtained by performing a pixel-wise k-means clustering with the cityblock distance on the local energy \eqref{eq:local_energy}, for~preselected numbers of clusters $k$.

Figure~\ref{fig:STMpartition} shows partitions of the Fourier spectrum of a scanning tunneling microscope image obtained by performing the Curvelet, Watershed and Voronoi methods with a scale-space step-size $s_0=0.2$ on the logarithm of the Fourier spectrum. The~Curvelet method results in a large number of sets $\Omega_n$, namely $393$, while the Watershed and Voronoi methods extract as many sets $\Omega_n$ as detected harmonic modes, $17$ in this~case.

\begin{figure}[!h] 
\centering
\includegraphics[width=.3\textwidth]{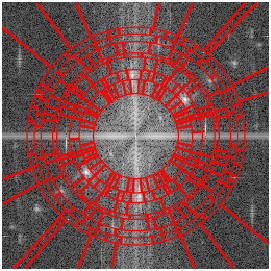} \includegraphics[width=.3\textwidth]{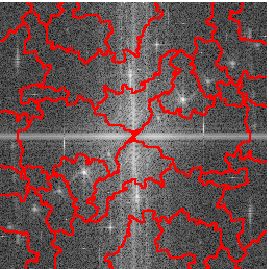} 
\includegraphics[width=.3\textwidth]{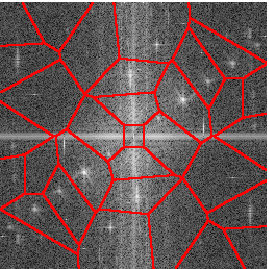} 
\caption{{\bf Examples of Fourier partitions.} Fourier spectrum of a scanning tunnelling microscope image overlapped with partitions obtained by the Curvelet (left), Watershed (middle) and Voronoi (right) methods.}\label{fig:STMpartition}
\end{figure}

Figure~\ref{fig:STMseg} shows scanning tunneling microscope images, their corresponding textural parts, and~their segmentation for the EWT-Curvelet and the proposed disk and square band-pass empirical wavelet transforms with $\tau=0.2$ for both Voronoi and Watershed partitioning. The~proposed empirical wavelet transforms provide accurate and comparable texture masks for the different images. In~all cases,  the~segmentation performance with the proposed approach is better than with the EWT-Curvelet, as~can be clearly seen in the last two rows of Figure~\ref{fig:STMseg}.

\begin{figure}[!ht] 
\centering
\includegraphics[width=1\textwidth]{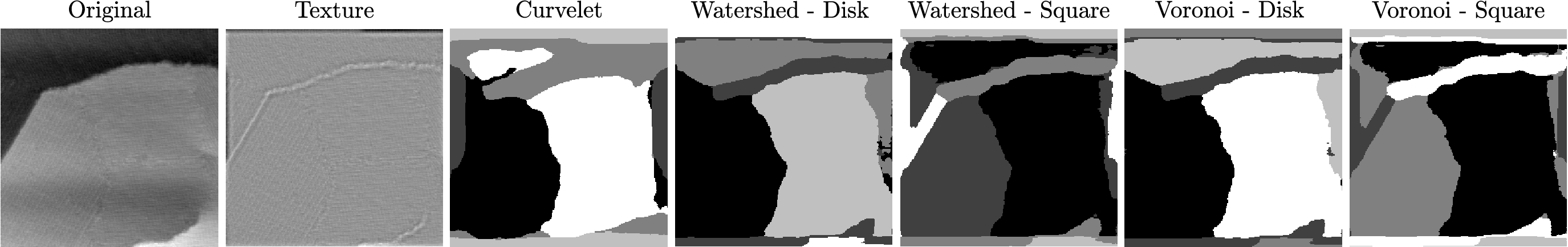}  \\
\includegraphics[width=1\textwidth]{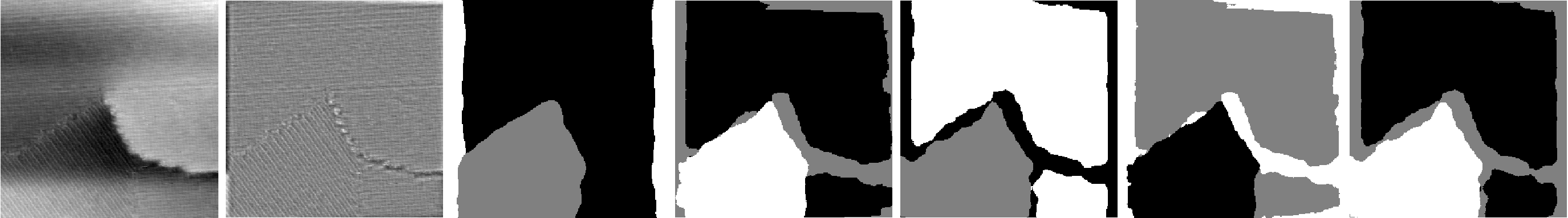}  \\
\includegraphics[width=1\textwidth]{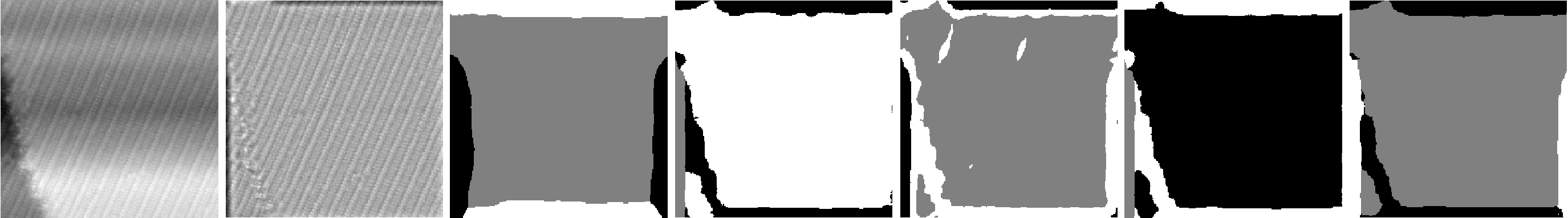}  \\
\includegraphics[width=1\textwidth]{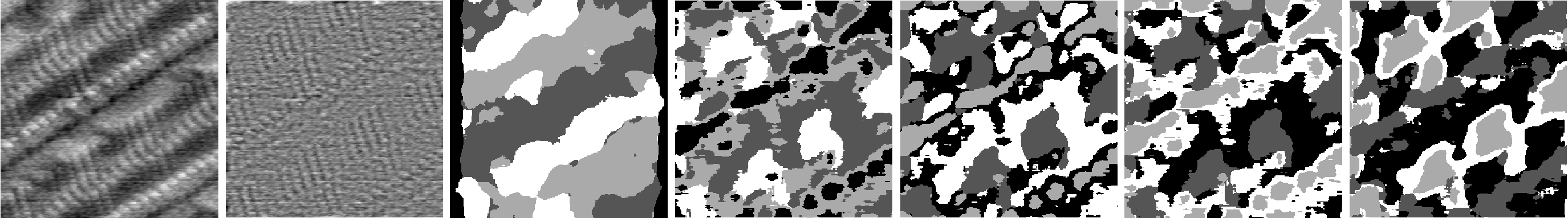}  \\
\caption{{\bf Texture segmentation.} Original scanning tunneling microscope images, their textural parts and the segmentation obtained for the different empirical wavelet transforms with numbers of clusters $k$ set to $5$, $3$, $3$ and $4$ from top to bottom, respectively. Raw scanning tunneling microscope images of cyanide on Au$\{111\}$ have been reproduced from \cite{guttentag2016hexagons,guttentag2016surface} with permission. Images copyright American Chemical Society.}\label{fig:STMseg}
\end{figure}

\section{Discussion}
\label{sec:discussion}

As shown by numerical experiments, the~additive demons algorithm provides accurate continuous mappings suitable for the empirical wavelet transform. The~results show that this approach has satisfactory performances for reconstruction and texture extraction. In~contrast, diffeomorphic mapping estimates are harder to obtain and take longer. This suggests that the use of normalization proposed in previous work~\cite{lucas2024multidimensional} should be avoided in~practice.

Moreover, although~the proposed algorithm applies to any homogeneous or separable wavelet kernel and any Fourier partition, the~selection of the wavelet kernel's support and the partitioning method is crucial to optimize the mapping estimation performance. The~study conducted in the present work shows that a Voronoi partition and a wavelet kernel defined on a disk result in a better performance~overall.

The tools introduced in this work are ready for application on real-world data. In~particular, they have been shown to provide relevant texture features for scanning tunneling microscope images for the different studied wavelet kernels. In~addition, the~number of features resulting from the Voronoi or Watershed partitioning methods is reasonable compared to the one obtained by the Curvelet transform used in~\cite{bui2020segmentation}. 

Future work will include a thorough study of the stop criterion of the additive demons algorithm. In addition, an~in-depth study of texture segmentation will also be carried out on a large labeled dataset, including a quantitative assessment of the impacts of the Fourier partitioning on the segmentation performance.

\section{Conclusions}
\label{sec:conclusions}

This work proposed an efficient algorithm for the empirical wavelet transform from any wavelet kernel based on demon registration. Among~the different meticulously compared demons algorithms, the~additive one has been shown to provide accurate numerical empirical wavelet systems with good properties of reconstruction, as~validated for suitable wavelet kernels. The~relevance of this approach for texture feature extraction has also been~highlighted.

\section{Acknowledgement}
This work was funded by the Air Force Office of Scientific Research, grant FA9550-21-1-0275.


\bibliographystyle{plain}
\bibliography{mybibfile}


%


\end{document}